\documentclass[11pt]{article}

\usepackage[margin=1in]{geometry}
\usepackage{amssymb,amsmath,amsfonts,amsthm}
\usepackage{cleveref}
\usepackage{graphicx}
\usepackage[ruled,vlined]{algorithm2e}
\usepackage{xcolor}
\usepackage{bbold}

\newenvironment{mechanism}[1][htb]
  {
   \begin{algorithm}[#1]%
  }{\end{algorithm}}

\crefname{ineq}{inequality}{inequalities}

\let\g\gamma

\let\e\varepsilon

\def\fe{\frac{1}{e}}
\let\cd\cdot

\newcommand{\reals}[0]{\mathbb{R}}
\newcommand{\vv}[0]{\vec{v}}

\newcommand{\vb}[0]{\vec{b}}
\newcommand{\vD}[0]{\vec{D}}

\newcommand{\vS}[0]{\vec{S}}
\newcommand{\vp}[0]{\vec{p}}

\newcommand{\mS}[0]{\mathcal{S}}
\newcommand{\mA}[0]{\mathcal{A}}

\newcommand{\sub}[0]{\subseteq}

\DeclareMathOperator*{\E}{\mathbb{E}}
\DeclareMathOperator*{\argmin}{\arg\min}

\def\sw{\mathrm{sw}}

\def\e{\varepsilon}
\def\S{\mathcal{S}}

\newtheorem{theorem}{Theorem}

\newtheorem{lemma}{Lemma}

\newtheorem{proposition}{Proposition}

\theoremstyle{definition}
\newtheorem{definition}{Definition}
\newtheorem{example}{Example}

\begin{document}

\title{Mechanism Design for Perturbation Stable Combinatorial Auctions
\thanks{This work was supported by the Hellenic Foundation for Research
and Innovation (H.F.R.I.) under the ``First Call for H.F.R.I. Research
Projects to support Faculty members and Researchers and the procurement
of high-cost research equipment grant'',  project BALSAM,
HFRI-FM17-1424.}}

\author{
  Giannis Fikioris\\
  School of Electrical and Computer Engineering\\
  National Technical University of Athens, 15780 Athens, Greece\\
  \texttt{fikioris@corelab.ntua.gr}
\and
  Dimitris Fotakis\\
  School of Electrical and Computer Engineering\\
  National Technical University of Athens, 15780 Athens, Greece\\
  \texttt{fotakis@cs.ntua.gr}
}

\maketitle

\begin{abstract} 
Motivated by recent research on combinatorial markets with endowed valuations by (Babaioff et al., EC~2018) and (Ezra et al., EC~2020), we introduce a notion of perturbation stability in Combinatorial Auctions (CAs) and study the extend to which stability helps in social welfare maximization and mechanism design. A CA is \emph{$\gamma$-stable} if the optimal solution is resilient to inflation, by a factor of $\gamma \geq 1$, of any bidder's valuation for any single item. On the positive side, we show how to compute efficiently an optimal allocation for $2$-stable subadditive valuations and that a Walrasian equilibrium exists for $2$-stable submodular valuations. Moreover, we show that a Parallel 2nd Price Auction (P2A) followed by a demand query for each bidder is truthful for general subadditive valuations and results in the optimal allocation for $2$-stable submodular valuations. To highlight the challenges behind optimization and mechanism design for stable CAs, we show that a Walrasian equilibrium may not exist for $\g$-stable XOS valuations for any $\g$, that a polynomial-time approximation scheme does not exist for $(2-\e)$-stable submodular valuations, and that any DSIC mechanism that computes the optimal allocation for stable CAs and does not use demand queries must use exponentially many value queries. We conclude with analyzing the Price of Anarchy of P2A and Parallel 1st Price Auctions (P1A) for CAs with stable submodular and XOS valuations. Our results indicate that the quality of equilibria of simple non-truthful auctions improves only for $\gamma$-stable instances with $\gamma \geq 3$. 

\medskip\medskip\noindent{{\bf Keywords:} Combinatorial Auctions, Perturbation Stability, Submodular Valuations, Price of Anarchy}
\end{abstract}

\section{Introduction}
\label{s:intro}

Combinatorial auctions appear in many different contexts (e.g., spectrum auctions \cite{Milgrom04}, network routing auctions \cite{HS01}, airport time-slot auctions \cite{RSB82}, etc.) and have been studied extensively (and virtually from every possible aspect) for a few decades (see e.g., \cite{CAbook} and the references therein).

In a combinatorial auction, a set $M$ of $m$ items (or goods) is to be allocated to $n$ bidders. Each bidder $i$ has a \emph{valuation} function $v_i : 2^M \to \reals_{\geq 0}$ that assigns a non-negative value $v_i(S)$ to any $S \sub M$ and quantifies $i$'s preferences over item subsets. Valuation functions are assumed to be non-decreasing (free disposal), i.e., $v_i(S) \geq v_i(S')$, for all $S' \subseteq S$, and normalized, i.e., $v(\emptyset) = 0$. The goal is to compute a partitioning (a.k.a. \emph{allocation}\,) $\mathcal{S} = (S_1, \ldots, S_n)$ of $M$ that maximizes the \emph{social welfare} $\sw(\S) = \sum_{i=1}^n v_i(S_i)$.

Most of the previous work has focused on CAs with either submodular (and XOS) or complement-free valuations. A set function $v: 2^M \to \reals_{\geq 0}$ is \emph{submodular} if for all $S, T \subseteq M$, $v(S) + v(T) \geq v(S \cap T) + v(S \cup T)$, and \emph{subadditive} (a.k.a. \emph{complement-free}\,) if $v(S) + v(T) \geq v(S \cup T)$. A set function $v$ is \emph{XOS} (a.k.a. \emph{fractionally subadditive}, see \cite{Feige}) if there are additive functions $w_k  : 2^M \to \reals_{\geq 0}$ such that for every $S \subseteq M$, $v(S) = \max_{k} \{ w_k(S) \}$. The class of submodular functions is a proper subset of the class of XOS functions, which in turn is a proper subset of subadditive functions.

Since bidder valuations have exponential size in $m$, algorithmic efficiency requires that the bidders communicate their preferences through either value or demand queries. A \emph{value query} specifies a bidder $i$ and a set (or bundle) $S \subseteq M$ and receives its value $v_i(S)$. A \emph{demand query} specifies a bidder $i$, a set $T$ of available items and a price $p_j$ for each available item $j \in T$, and receives a bundle $S \subseteq T$ that maximizes $i$'s \emph{utility} $v_i(S) - \sum_{j \in S} p_j$ from the set of available items at these prices. 
Demand queries are strictly more powerful than value queries, in the sense that value queries can be simulated by polynomially many demand queries, and in terms of communication cost, demand queries are exponentially stronger than value queries \cite{BN09}.

The approximability of social welfare maximization by polynomial-time algorithms and truthful mechanisms for CAs with submodular and subadditive bidders has been extensively studied by the communities of Approximation Algorithms and Algorithmic Mechanism Design in the last two decades and are practically well understood (see e.g., Section~\ref{s:previous} for a selective list of references most relevant to our work).

\subsection{Perturbation Stability in Combinatorial Auctions}

Motivated by recent work on beyond worst-case analysis of algorithms \cite{Rough19,Rough20} and on endowed valuations for combinatorial markets \cite{BDO18,EFF20}, in this work, we investigate whether strong performance guarantees for social welfare maximization (by polynomial-time algorithms and truthful mechanisms, or even at the equilibrium of simple auctions) can be achieved for a very restricted (though still natural) class of CAs with \emph{perturbation stable} valuations, where the optimal solution is resilient to a small increase of any bidder's valuation for any single item. 

From a bird's view, we follow the approach of \emph{beyond worst-case analysis} (see e.g., \cite{Rough20,Rough19}), where we seek a theoretical understanding of the superior practical performance of certain algorithms by formally analyzing them on practically relevant instances. Hence, researchers restrict their attention to instances that satisfy certain application-area-specific assumptions, which are likely to be satisfied in practice. Such assumptions may be of stochastic (e.g., smoothed analysis of Simplex and local search \cite{ST04,ST09,ERV16}) or deterministic nature (e.g., perturbation stability in clustering \cite{AMM17,ABS12,BHW16,BL10}).

The beyond worst-case approach is not anything new for (Algorithmic) Mechanism Design. \emph{Bayesian} analysis, where the bidder valuations are drawn as independent samples from an arbitrary distribution known to the mechanism, is standard in revenue maximization \cite{RoughOpt} and has led to many strong and elegant results for social welfare maximization by truthful posted price mechanisms (see e.g., \cite{FGL14,DFKL17}). However, in this work, we significantly deviate from Bayesian analysis, where the mechanism has a relatively accurate knowledge of the distribution of bidder valuations. Instead, we suggest a deterministic restriction on the class of instances (namely, perturbation stability) and investigate if there is a natural class of mechanisms (e.g., P2A) that are incentive-compatible and achieve optimality for CAs with stable submodular valuations.

Our focus on perturbation stable valuations was actually motivated by the recent work of Babaioff et al. \cite{BDO18} and Ezra et al. \cite{EFF20} on combinatorial markets where the valuations exhibit the endowment effect. The \emph{endowment effect} was proposed by the Nobel Laureate Richard Thaler \cite{thaler_2000} to explain situations where owning a bundle of items causes its value to increase. Babaioff et al. \cite{BDO18} defined that if an allocation $\S = (S_1, \ldots, S_n)$ is $\alpha$-endowed, for some $\alpha > 1$, in a CA with bidder valuations $(v_1, \ldots, v_n)$, then the valuation function of each bidder $i$ becomes
\begin{equation}\label{eq:endowed}
 v'_i(T) = v_i(T) + (\alpha - 1) v_i(S_i \cap T)\,,
\end{equation}
for all item sets $T \subseteq M$. Roughly speaking, the value of $S_i$ (and its subsets) is inflated by a factor of $\alpha$ due to the endowment effect. The main result of \cite{BDO18} is that for any combinatorial market with submodular valuations $(v_1, \ldots, v_n)$, any locally optimal allocation $\S$ and any $\alpha \geq 2$, the market with $\alpha$-endowed valuations $(v'_1, \ldots, v'_n)$ for $\S$ admits a Walrasian equilibrium (see Section~\ref{s:prelims} for the definition) where each bidder $i$ receives $S_i$. In simple words, social welfare maximization in combinatorial markets with endowed valuations $(v'_1, \ldots, v'_n)$ is polynomially solvable and the optimal allocation is supported by item prices. Subsequently, Ezra et al. \cite{EFF20} presented a general framework for endowed valuations and extended the above result to XOS valuations and general valuations, for a sufficiently large endowment (see also previous work on bundling equilibrium and conditional equilibrium \cite{DFTW18,FKL12}).

Inflated valuations due to the endowment effect naturally occur in auctions that take place regularly over time. Imagine auctions for e.g., season tickets of an athletic club, spots in a parking lot, reserving timeslots for airport gates, vacation packages at resorts, etc., where regular participants tend to value more the bundles allocated to them in the past, due to the endowment effect (see also \cite{thaler_2000} for more examples). Given the strong positive results of \cite{BDO18,EFF20}, a natural question is whether CAs with valuations inflated due to the endowment effect allow for stronger approximation guarantees in social welfare maximization and mechanism design.

\smallskip\noindent{\bf Stable Combinatorial Auctions.}
To investigate the question above, we adopt a slightly stronger condition on valuation profiles, namely perturbation stability, which is inspired by (and bears a resemblance to) the definition of perturbation stable clustering instances (see e.g., \cite{AMM17,ABS12,BHW16,BL10}).

\begin{definition}\label{def:stable}
For a constant $\gamma \geq 1$, a \textit{$\gamma$-perturbation} of a valuations profile $\vec{v} = (v_1, \ldots, v_n)$ on a bidder $i$ and an item $j$ is a new valuations profile $\vec{v}' = (v'_1, \ldots, v'_n)$, where $v'_k = v_k$ for all bidders $k \neq i$, and for all $S \subseteq M$,
\begin{equation}\label{eq:perturbation}
v'_i(S) = v_i(S) + (\gamma - 1) v_i(S \cap \{ j \})
\end{equation}
A CA with valuations $\vec{v} = (v_1, \ldots, v_n)$ is \emph{$\gamma$-perturbation stable} (or \emph{$\gamma$-stable}) if the optimal allocation for $\vec{v}$ is unique and remains unique for all $\gamma$-perturbations $\vec{v}'$ of $\vec{v}$.
\end{definition}

\begin{example}\label{ex:AliceBob}
    Let Alice and Bob compete for 2 items, $a$ and $b$, and have valuations $v_A(\{a\}) = v_A(\{a, b \}) = 2$ and $v_A(\{ b \}) = 1$, and $v_B(\{b\}) = v_B(\{a, b \}) = 2$ and $v_B(\{ a \}) = 1$. The (unique) optimal allocation is to give $a$ to Alice and $b$ to Bob, with social welfare $4$. A perturbation with most potential to change the optimal solution is to inflate Alice's value of $b$ by $\gamma \geq 1$. Then, we get $v'_A(\{a\}) = 2$, $v'_A(\{ b \}) = \gamma$ and $v'_A(\{a, b \}) = 1+\gamma$. The optimal solution remains unique for any $\gamma < 3$. Hence the above CA is $(3-\e)$-stable, for any $\e > 0$.\qed
\end{example}

At the conceptual level, we feel that the condition of $\gamma$-stability is easier to grasp and to think about in the context of mechanism design for CAs (compared against considering valuation profiles $\vec{v}$ resulting from the $\alpha$-endowment of an optimal solution to an initial valuations profile $\vec{x}$)%
\footnote{\label{foot:property}For a better understanding of the two conditions at a technical level, we note that a (technically very useful) necessary condition for a valuations profile $\vec{v}$ to be $\gamma$-stable is that for the optimal allocation $(O_1, \ldots, O_n)$, any bidders $i \neq k$ and any item $j \in O_i$,
\[  v_i(O_i) - v_i(O_i \setminus \{ j \}) > v_k(O_k \cup \{ j \}) - v_k(O_k) + (\gamma - 1) v_k(\{ j \}) \geq (\gamma - 1) v_k(\{ j \}) \,. \]
For this condition, we use (local) optimality of $(O_1, \ldots, O_n)$ for both $\vec{v}$ and its $\gamma$-perturbation on bidder $k$ and item $j$ (see also Lemma~\ref{LemmaGen}).

A similar (technically useful) condition satisfied by any valuations profile $\vec{v}$ that has resulted from the $\alpha$-endowment of an optimal  (or locally optimal) solution $(O_1, \ldots, O_n)$ to an initial valuations profile $\vec{x}$ is that for any bidders $i \neq k$ and any item $j \in O_i$,
\[  v_i(O_i) - v_i(O_i \setminus \{ j \}) \geq \alpha \big(v_k(O_k \cup \{ j \}) - v_k(O_k) \big)\,. \]
For this condition, we use local optimality of $(O_1, \ldots, O_n)$ for $\vec{x}$, multiply the resulting inequality by $\alpha$, and observe that $v_i(O_i) - v_i(O_i \setminus \{ j \}) = \alpha\big(x_i(O_i) - x_i(O_i \setminus \{ j \})\big)$ and that $v_k(O_k \cup \{ j \}) - v_k(O_k) = x_k(O_k \cup \{ j \}) - x_k(O_k)$.}.
From an algorithmic and mechanism design viewpoint, we remark that for any $\gamma \geq 2$, CAs with $\gamma$-stable submodular valuations can be treated (to a certain extent) as multi-item auctions with additive bidders. In fact, this is the technical intuition behind several of our positive results.

\subsection{Contributions}
\label{s:contrib}

We focus on deterministic algorithms and mechanisms. We first show that a simple greedy algorithm (Algorithm~\ref{AlgoSubA}) that allocates each item $j$ to the bidder $i$ with maximum $v_i(\{ j \})$ finds the optimal allocation for CAs with $2$-stable subadditive valuations (Theorem~\ref{thm:OptSubA}). Moreover, similarly to \cite{BDO18}, we show that for $2$-stable submodular valuations, combining the optimal allocation with a second price approach, where each item $j$ gets a price of $p_j = \max_{k\neq i} v_k(\{ j \})$, results in a Walrasian equilibrium (Theorem~\ref{thm:WalEq}).

On the negative side, we prove that our positive results above cannot be significantly strengthened. We first show that there is a simple $(2-\e)$-stable CA with submodular bidders where approximating the social welfare within any factor larger than $1-\frac{1}{2k}$ requires at least $\binom{m}{k}$ value queries, for any integer $k \geq 1$ (Theorem~\ref{thm:Under2Stable}). Thus, a polynomial-time approximation scheme does not exist for $(2-\e)$-submodular valuations. Moreover, we show that for any $\gamma \geq 1$, there is a $\gamma$-stable CA with a XOS bidder and a unit demand bidder that does not admit a Walrasian equilibrium (Lemma~\ref{lem:no_WE_XOS}).

On the mechanism design part, in a nutshell, we show that (possibly appropriately modified) Parallel 2nd Price Auctions (P2A) behave very well for stable CAs. We should highlight that despite the fact that maximizing the social welfare for $2$-stable subadditive CAs is easy, VCG is not an option for the design of computationally efficient incentive compatible mechanisms. The reason is that removing a single bidder from a $2$-stable CA may result in an  NP-hard (and hard to approximate) (sub)instance.

In Section~\ref{s:mechanisms}, we show that a P2A followed by a demand query for each bidder is dominant strategy incentive compatible (DSIC) for all CAs with subadditive bidders and maximizes the social welfare if the valuations profile is submodular and $2$-stable (Theorem~\ref{thm:OptSubM}). If demand queries are not available, the mechanism boils down to a simple P2A. We show that P2A is ex-post incentive compatible (EPIC) for $2$-stable submodular valuations and that truthful bidding leads to the optimal allocation.

On the negative side and rather surprisingly, we show that demand queries are indeed necessary for computationally efficient mechanisms that are DSIC for all submodular valuations and maximize the social welfare if the instance is $\gamma$-stable (even if $\gamma$ is arbitrarily large, Theorem~\ref{thm:DSIC_no_exist}). Our construction is an insightful adaptation of the elegant lower bound in \cite[Theorem~1]{Dob11} to the case of stable submodular valuations. We show that any DSIC mechanism that computes the optimal allocation for stable CAs and does not use demand queries must use exponentially many value queries. The crux of the proof is that in certain instances, the bidders may find profitable to misreport and switch from a non-stable instance to a stable one.

In Section~\ref{s:poa}, we analyze the Price of Anarchy (PoA) of P2A and Parallel 1st Price Auctions (P1A). Our results demonstrate that the quality of equilibria of simple non-truthful auctions improves only for $\gamma$-stable valuations, with $\gamma \geq 3$. We show that the PoA of P2A for CAs with $3$-stable submodular valuations is $1$ (Theorem~\ref{theoremP2AStable}), while there are $(3-\e)$-stable CAs with PoA equal to $1/2$ (Lemma~\ref{lem:PoA_counter_example}), which matches the PoA for CAs with general submodular valuations (see e.g., \cite{RST16}). Moreover, we show that the PoA of both P2A and P1A for CAs with $\gamma$-stable XOS valuations is at least $\frac{\g-2}{\g-1}$, for any $\gamma \geq 2$ (Theorem~\ref{theoremStableP2AXOS} and Theorem~\ref{theoremStableP1A}).

\subsection{Previous Work}
\label{s:previous}

Social welfare maximization with submodular and subadditive valuations has been studied extensively. Submodular Welfare Maximization (SMOD-WM) is known to be $(1-1/e)$-approximable with polynomially many value queries \cite{VondrakStoc08} and $(1-1/e+\e)$-approximable, for a fixed constant $\e > 0$, with polynomially many demand queries \cite{FVFocs06}. Moreover, a simple and natural greedy algorithm achieves an approximation ratio of $1/2$ using only value queries \cite{LLN06}. The results about polynomial-time approximability with value queries are best possible, in the sense that approximating SMOD-WM within a factor of $1-1/e+\e$, for any constant $\e > 0$, is NP-hard \cite{KLMM05} and requires exponentially many value queries \cite{MSV07}. Furthermore, there is a constant $\e > 0$, such that approximating SMOD-WM within a factor of $1-\e$ with demand queries is NP-hard \cite{FVFocs06}. Subadditive Welfare Maximization (SADD-WM) is $m^{-1/2}$-approximable with polynomially many value queries (and this is best possible \cite{MSV07}) and $1/2$-approximable with polynomially many demand queries \cite{Feige}.

Truthful maximization of social welfare in CAs with submodular (or XOS) bidders has been a central problem in Algorithmic Mechanism Design.
In the \emph{worst-case} setting, where we do not make any further assumptions on bidder valuations, Dobzinski et al. \cite{DNS06} presented the first truthful mechanism that uses polynomially many demand queries and achieves a non-trivial approximation guarantee of $O((\log m)^{-2})$. Dobzinski \cite{D07} improved the approximation ratio to $O(\frac{1}{\log m \log \log m})$ for the more general class of subadditive valuations. Subsequently, Krysta and V\"{o}cking \cite{KV12} provided an elegant randomized online mechanism with an approximation ratio of $O(\frac{1}{\log m})$ for XOS valuations. Dobzinski \cite{D16} broke the logarithmic barrier for XOS valuations, by showing an approximation guarantee of $O((\log m)^{-1/2})$, which was recently improved to $O((\log\log m)^{-3})$ by Assadi and Singla \cite{AS19}. Accessing valuations through demand queries is essential for these strong positive results. Dobzinski \cite{Dob11} proved that any truthful mechanism for CAs with submodular bidders with approximation ratio better than $m^{-\frac{1}{2}+\e}$ must use exponentially many value queries. Truthful $\Theta(m^{-1/2})$-approximate mechanisms that use polynomially many value queries are known even for the more general class of subadditive valuations (see e.g., \cite{DNS06}).

In the Bayesian setting, Feldman et al. \cite{FGL14} showed how to obtain item prices that provide a constant approximation ratio for XOS valuations. These results were significantly extended and strengthened by D{\"u}etting et al. \cite{DFKL17}.

Previous work has also shown strong Price of Anarchy (PoA) guarantees for CAs with submodular, XOS and subadditive bidders that can be achieved by simple (non-truthful) auctions, such as P2A and P1A (see e.g., \cite{CKS08,ST14,RST16,Rough14}).

Our notion of perturbation stability for CAs is inspired by conceptually similar notions of perturbation stability in clustering \cite{ABS12,BL10}. Angelidakis et al. \cite{AMM17} presented a polynomial-time algorithm for $2$-stable clustering instances with center-based objectives (e.g., $k$-median, $k$-means, $k$-center), while Balcan et al. \cite{BHW16} proved that there is no polynomial-time algorithm for $(2-\e)$-stable instances of $k$-center, unless NP = RP. To the best of our knowledge, this is the first time that the notion of perturbation stability has been applied to social welfare maximization and to algorithmic mechanism design for Combinatorial Auctions.

\section{Notation and Preliminaries}
\label{s:prelims}

The key notion of $\gamma$-perturbation stability (Definition~\ref{def:stable}) and a significant part of the terminology and the notation are introduced in Section~\ref{s:intro}. In this section, we introduce some additional terminology, notation and conventions used in the technical part. 

We always let $\mathcal{O} = (O_1, \ldots, O_n)$ denote the optimal allocation for the instance at hand, and let $O_i$ be the bundle of bidder $i$ in $\mathcal{O}$. For convenience, we usually let an index $j$ also denote the singleton set $\{ j \}$ (we write $v_i(j)$, $v_i(S \cup j)$, $v_i(S \setminus j)$, etc., instead of $v_i(\{j\})$, $v_i(S \cup \{j\})$, $v_i(S \setminus \{j\})$). We use both $S_1 \setminus S_2$ and $S_1 -S_2$ for the set difference. We denote the marginal contribution of a bundle $S$ wrt. $T$ as $v(S|T) = v(S\cup T) - v(T)$.

In addition to submodular, XOS, and subadditive valuations, we consider the classed of \emph{additive} and \emph{unit-demand} valuations $v: 2^M \to \reals_{\geq 0}$, where there exist $b_1,\ldots,b_m \in \reals_{\geq 0}$, such that for any $S \sub M$, $v(S)=\sum_{j\in S} b_j$ and $v(S)=\max_{j\in S} b_j$, respectively. A useful property of an XOS valuation $v$ is that for any $S \sub M$, there is an additive valuation $q$ that \emph{supports} $S$, in the sense that $v(S) = q(S)$ and for any $T \sub M$, $v(T) \geq q(T)$.

We focus on deterministic algorithms and mechanisms and consider bidders with \emph{quasi-linear utilities}, where the utility of bidder $i$ with valuation $v_i$ for a bundle $S$ at price $p(S)$ is $u_i(S) = v_i(S) - p(S)$. For a price vector $(p_1,\ldots,p_m)$, we often let $p(S) = \sum_{j \in S} p_j$ denote the price of a bundle $S \sub M$. 

An allocation $\S = (S_1, \ldots, S_n)$ and a price vector $(p_1,\ldots,p_m)$ form a \emph{Walrasian Equilibrium} if all items $j$ with $p_j > 0$ are  allocated and each bidder $i$ gets a utility maximizing bundle (or, his \emph{demand}) in $\S$, i.e., $\forall S \sub M$, $v_i(S_i) - p(S_i)\geq v_i(S) - p(S)$. 

A mechanism is \emph{dominant-strategy incentive compatible} (DSIC) (or \emph{truthful}) if for any valuations profile $\vec{v}$, answering (value or demand) queries truthfully is a dominant strategy and guarantees non-negative utility for all bidders. A mechanism is called \emph{ex-post incentive compatible} (EPIC) if truthful bidding is an ex-post Nash equilibrium and guarantees non-negative utility for all bidders.

Let $\vD = (D_1,..., D_n)$ be a profile of distributions over valuation functions (i.e., over possible bids). In a mechanism with allocation rule $\vS(\cd) = (S_1(\cd),\ldots,S_n(\cd))$ and item pricing rule $\vp = (p_1(\cd),\ldots,p_m(\cd))$, $\vD$ forms a \emph{Mixed Nash Equilibrium} (MNE), if no bidder has an incentive to unilaterally deviate from $\vD$, i.e., for any bidder $i$ and any distribution $D_i'$ over valuation functions, 
   \begin{equation*}
      \E_{\vb \sim \vD} \left[v_i(S_i(\vb)) - \sum_{j\in S_i(\vb)}p_j(\vb)\right] \geq 
      \E_{\vb \sim (D_i', \vD_{-i})} \left[v_i(S_i(\vb)) - \sum_{j\in S_i(\vb)}p_j(\vb)\right]
   \end{equation*}
If instead of distributions over valuation functions, we restrict each $D_i$ and $D_i'$ to valuation functions (i.e., to pure strategies over possible bids), we get the definition of a \emph{Pure Nash Equilibrium} (PNE). 

In Section~\ref{s:poa}, we consider the \emph{Price of Anarchy} (PoA) of Parallel 2nd Price Auctions (P2A) and Parallel 1st Price Auctions (P1A). In a Combinatorial Auction, the PoA of a mechanism is the ratio of (resp. expected) social welfare at the worst Pure (resp. Mixed) Nash Equilibrium to the social welfare of the optimal allocation. Formally, focusing on the more general case of Mixed Nash Equilibria: 
\begin{equation*}
      \text{PoA} = \min_{\vD \text{ is a MNE}} \frac{\E_{\vb\sim\vD}\left[\sum_i v_i(S_i(\vb))\right]}{\sum_i v_i(O_i)}
   \end{equation*}

\noindent{\bf Properties of Stable Valuations.}
The following shows a technically useful property of $\gamma$-stable CAs (see also Footnote~\ref{foot:property}).

\begin{lemma}[Valuation Stability] \label{LemmaGen}
    Let $\vv$ be $\g$-stable and subadditive valuations. Then, for all bidders $i\neq k$ and all items $j \in O_i$
    \begin{equation*}
        v_i(j) \geq v_i(O_i) - v_i(O_i \setminus j) > (\g-1) v_k(j)
    \end{equation*}
\end{lemma}

\begin{proof}
    Fix a bidder $i$, an item $j\in O_i$ and another bidder $k \neq i$. Since the valuations are $\g$-stable, the optimal allocation remains unique and optimal, even if we inflate $k$'s value for item $j$ by $\g$. Such a $\gamma$-perturbation of $v_k$ results in $v'_k(O_k) = v_k(O_k)$ and $v'_k(O_k \cup j) = v_k(O_k \cup j) + (\g-1) v_k(j)$. Using the optimality of the allocation $(O_1, \ldots, O_n)$ after the perturbation, we obtain that: 
    \begin{equation*}
    v_i(O_i) - v_i(O_i \setminus j) > v_k'(O_k\cup j) - v_k(O_k) = v_k(O_k \cup j)  - v_k(O_k) + (\g-1) v_k(j)
    \end{equation*}
    Using that $v_k$ is non-decreasing concludes the proof of the lemma.
\end{proof}

\section{Social Welfare Maximization for Stable Valuations} 
\label{sec:Optimization}

We next consider the problem of social welfare maximization for $2$-stable CAs. We first show that for $2$-stable subadditive valuations, we can compute the optimal solution with value queries in polynomial time. 

\begin{algorithm}[tb]
\caption{Algorithm for $2$-Stable Subadditive Valuations}
\label{AlgoSubA}
\SetKwInOut{Input}{Input}
\Input{Value query access to subadditive valuations $v_1(\cd),...,v_n(\cd)$}
\smallskip

    Set $O_1 = O_2 = ... = O_n = \emptyset$ \\
    \For{$j \in M$}{
        Let $i$ be the bidder that maximizes $v_i(j)$, and set $O_i \leftarrow O_i\cup\{j\}$. \\
    }
    \Return{Allocation $(O_1,...O_n)$} \\
\end{algorithm}

\begin{theorem}\label{thm:OptSubA}
Let $\vv$ be a $2$-stable subadditive valuations profile. Then Algorithm~\ref{AlgoSubA} outputs the optimal allocation $(O_1,...,O_n)$ using $nm$ value queries.
\end{theorem}

\begin{proof}
The number of value queries follows directly from the description of the algorithm. As for optimality, we fix an item $j$ and let $i$ be the bidder that gets $j$ in the optimal solution. Because of Lemma~\ref{LemmaGen} and the fact that $\vv$ is $2$-stable, we know that $v_i(j)>v_k(j)$, for any other bidder $k\neq i$. Because Algorithm~\ref{AlgoSubA} allocates $j$ to the bidder with the highest singleton value, $i$ gets item $j$ in the allocation of Algorithm~\ref{AlgoSubA}.
\end{proof}

On the negative side, we next show that a polynomial-time approximation scheme does not exist even for $(2-\e)$-stable submodular valuations.

\begin{theorem}\label{thm:Under2Stable}
    For any $\e > 0$, there exists a submodular $(2-\e)$-stable valuations profile $\vv$ such that for any integer $k \geq 1$, approximating the optimal allocation in $\vv$ within any factor larger than $1 - \frac{1}{2k}$ requires at least $\binom{m}{k}$ value queries.
\end{theorem}

\begin{proof}
    Given a set $O\sub M$ we define the following class of valuations:
    \begin{equation*}
        v^O(S) =
            \begin{cases}
                |S|, &\textit{if } |S| \leq |O|-1 \\
                |O|-1/2, &\textit{if } |S| = |O| \textit{ and } S\neq O\\
                |O|, &\textit{otherwise ($S=O$ or $S \geq |O| + 1$)}\\
            \end{cases}
    \end{equation*}

    We first prove that the above valuation is submodular for any $O$. To do that we examine the value of $v^O(S\cup j) - v^O(S)$, given that $j\notin S$:
    \begin{equation*}
        v^O(S\cup j) - v^O(S) =
            \begin{cases}
                1, &\textit{if } |S| \leq |O|-2 \\
                1/2 \textit{ or } 1, &\textit{if } |S| = |O|-1\\
                0 \textit{ or } 1/2, &\textit{if } |S| = |O|\\
                0, &\textit{if } |S| \geq |O| + 1\\
            \end{cases}
    \end{equation*}
    This entails that for any sets $S,T$, where $|S|\leq |T|$, it holds that $v^O(S|j)\geq v^O(T|j)$. This proves the submodularity, for any $O\sub M$.

    Now we note that given a valuation $v^O(\cd)$ and value query access to it, finding the bundle $O$ requires at least $\binom{m}{|O|} - 1$ value queries: Let $O'$ be such that $|O|=|O'|$. Observe that the valuations $v_O$ and $v_O'$ differ only in their value for $O$ and $O'$. Thus, a query for the value of a bundle $S$ only tells us whether the valuation is $v^S$ or not. In the worst case, we have to query the value of every bundle $S$ for which $|S|=|O|$ (except the ``last'' bundle) to determine $O$.

    Now we fix the number of items $m$ and the number of bidders $n$ such that $m/n=k$ and $k$ is an integer. We also fix a partition of the items $(O_1,..., O_n)$, where $|O_1| = ... = |O_n| = k$ and the valuation of each bidder $i$ to $v^{O_i}(\cd)$.

    First we note that allocation $(O_1,...,O_n)$ is the optimal allocation, with welfare $m$, because a higher welfare is impossible, as each bidder values $M$ for $m/n$.

    Now we will prove that the valuations are $(2-\e)$-stable. Notice that the allocation with the second highest welfare has welfare $m-1$: A welfare of $m - 1/2$ is not achievable as one bidder would need to contribute $k-1/2$ and all the others exactly $k$. To do that we would have to give each bidder exactly $k$ items and exactly one bidder would need to not take his optimal set. This entails that a second bidder would also not get his optimal, thus making both bidders contribute $k - 1/2$ to the total welfare.

    Since the second best allocation has welfare $m-1$ and every singleton value is 1, the maximum value that the second highest allocation can achieve even with inflation for 1 item by $\g$ is $m - 1 + (\g-1)\cd 1$. Thus the valuations are $\g$-stable if $m > m - 1 + (\g-1)\cd 1$. This entails that $\g$ can take any value less 2.

    Because of the structure of the valuations, finding the optimal allocation is hard. More specifically, in the worst case allocating any bidder $i$ his correct bundle $O_i$ requires $\binom{m}{k}$ value queries. This means that with less queries no bidder is going to get his optimal set. If no bidder gets his optimal set the highest achievable welfare is guaranteed when all the bidders get $k$ items. This leads to an approximation ratio
    \begin{equation*}
        \frac{n\cd(k-1/2)}{m} = 1 - \frac{1}{2k}
    \end{equation*}

    Thus we have proven that the inapproximability. Together with the submodularity and the stability of the valuations, this concludes the proof. 
\end{proof}



\section{Existence of Walrasian Equilibrium}
\label{s:walras}

Similarly to \cite[Theorem~4.2]{BDO18}, we next show that combinatorial markets with $2$-stable submodular valuations admit a Walrasian Equilibrium. 

\begin{theorem}\label{thm:WalEq}
    Let $\vv$ be $2$-stable submodular valuations. For every bidder $i$ every item $j\in O_i$, let $p_j$ be the price of that item for which: $\max_{k\neq i}v_k(j) \leq p_j \leq v_i(j|O_i-j)$. Then, the prices $p_1,\ldots,p_m$ form a Walrasian Equilibrium.
\end{theorem}

\begin{proof}
   Fix a bidder $i$ and his optimal bundle $O_i$. We note that the price $p_j$ of each item is well defined. Because of Lemma~\ref{LemmaGen}, 2-stability, and $j\in O_i$, $v_i(j|O_i-j) > v_k(j)$.
   
    We next show that for any $j\not\in O_i$, bidder $i$ is not interested in getting $j$. Because of submodularity, $i$'s additional utility due to item $j$ is at most $v_i(j) - p_j$. Because $j\not\in O_i$, $p_j \geq v_i(j)$, making $i$'s utility from $j$ non-positive. Hence, $i$'s demand is a subset of $O_i$.
    
    To conclude the proof, we show that for any item $j\in O_i$, $i$ gets non-negative utility due to $j$. Fix a bundle $S$ in the demand of bidder $i$ with $j \not\in S$. Note that we have already proven that $S\sub O_i$. The utility gained by taking $j$ is $v_i(j|S) - p_j$, which is non-negative; submodularity makes $v_i(j|S) \geq v_i(j|O_i-j) \geq p_j$. Hence, $O_i$ is the demand of bidder $i$.
\end{proof}


We next show that Theorem~\ref{thm:WalEq} cannot be extended to stable XOS valuations. In the proof of Theorem~\ref{thm:WalEq}, $2$-stability and submodularity ensure that the prices cannot exceed the marginal increase $v_i(S|j)$ of adding an item $j$ to some $S \subset O_i$. For XOS valuations, however, the marginals $v_i(S|j)$ may not be increasing with $S$. As a result, the utility of a bidder $i$ may be maximized by a strict subset of his optimal bundle $O_i$. This is the core idea for the proof of the following lemma.

\begin{lemma}\label{lem:no_WE_XOS}
    For every $\g\geq 1$, there exists a $\gamma$-stable valuations profile with a XOS bidder and a unit demand bidder which does not admit a Walrasian Equilibrium.
\end{lemma}

\begin{proof}
    Fix a value of $\g$ and suppose there are $m$ items, where $m > \g + 2$. The XOS bidder's valuation is simply a function of the carnality of his bundle:
    \begin{equation*}
        v_1(S) =
            \begin{cases}
                0, &\textit{if }\; |S| = 0 \\
                X, &\textit{if }\; 1 \leq |S| \leq m-1 \\
                X + \g + 1, &\textit{if }\; |S| = m \\
            \end{cases}
    \end{equation*}
    where $X$ is a very large positive value. It is easy to verify that the above valuation is indeed monotone. To prove that it is also XOS we arbitrarily order the items $M=\{1,2,...,m\}$ and use the following $m+1$ additive functions
    \begin{itemize}
        \item For $j\in[m]$, $a_j(\cd)$, values item $j$ for $X$, and the rest of the items for 0, i.e. $\forall S$: $a_j(S) = X\cd\mathbb{1}(j\in S)$.
        \item The $(m+1)$-th function, $a_{m+1}(\cd)$, values every item for $\frac{X+\g+1}{m}$, i.e. $\forall S$: $a_{m+1}(S) = \frac{X+\g+1}{m}|S|$.
    \end{itemize}

    We now notice that $v_1(S) = \max_{k\in [m+1]}\sum_{t\in S}a_k(t)$. We should add that this holds only if the value of $X$ is large enough.

    The unit-demand bidder values each item for $1$: $v_2(S) = \mathbb{1}(|S|>0)$. The optimal solution is to give the whole bundle to the XOS bidder. If we endow the valuation of the unit-demand bidder by $\g$ for any item $j$, it would make it $v_2'(S) = \max\big(\mathbb{1}(|S|>0),\; \g\cd\mathbb{1}(j\in S)\big)$. Even if the unit demand bidder had this valuation, then the optimal allocation still is to give the XOS bidder all the items, making the valuations $\g$-stable.

    Now we are going to prove that this instance does not admit a Walrasian Equilibrium. To do that let us assume that there exist prices $p_1,..,p_m$ that form a Walrasian Equilibrium. Denote with $j$ the item whose price is minimum, i.e. $j\in\argmin_{j\in M}p_j$. The XOS bidder's utility for the cheapest item should be at most his utility for the grand bundle:
    \begin{equation*}
        X + \g + 1 - \sum_{t\in M}p_t \geq X - p_j
    \end{equation*}
    By reordering and using the fact that $j$ has the minimum price we get
    \begin{equation}\label{eq:XOS_utility}
        \g + 1 \geq (m-1) p_j
    \end{equation}

    Now we also use the fact the the utility of the unit-demand bidder for taking item $j$ should not be strictly positive, since we have a WE
    \begin{equation}\label{eq:unit_demand_utility}
        1 - p_j \leq 0
    \end{equation}

    By combining equations \ref{eq:XOS_utility} and \ref{eq:unit_demand_utility} we get $\g + 1 \geq m - 1$. Due to our original assumption that $m > \g + 2$, we get a contradiction which concludes the proof.     
\end{proof}

\section{Mechanism Design For Stable Combinatorial Auctions}
\label{s:mechanisms}

In this section, we investigate truthful mechanism design for CAs with stable submodular valuations. We should emphasize that despite Theorem~\ref{thm:OptSubA}, VCG cannot be used as a computationally efficient DSIC mechanism for stable subadditive CAs, because the subinstances $\vec{v}_{-i}$, whose optimal solutions determine the payments, may not be stable and may be NP-hard to solve optimally (e.g., adding a bidder with additive valuation that has a huge value for each singleton to any CA results in a stable instance).

We first present a truthful extension of Algorithm~\ref{AlgoSubA}, which is implemented as a Parallel 2nd Price Auction (P2A) and also uses a demand query for each bidder. 

\begin{mechanism}[t]
\label{mech:DSIC}
\caption{Extended Parallel 2nd Price Auction (EP2A)}

\SetKwInOut{Input}{Input}

    \Input{Value and demand query access to valuations $\vv = (v_1, \ldots, v_n)$.}
    \smallskip
    
    For all bidders $i$ and items $j$, query $v_i(j)$ and let $b_{ij}$ denote the response.\\
    Set the price $p_j$ of each item $j$ to its second highest bid.\\
    For each bidder $i$, let $S_i$ be the set of items for which $i$ has the highest bid.\\
    Bidder $i$ receives his demand from $S_i$, where each item has price $p_j$. \\
    Bidder $i$ pays the total price for his demand from $S_i$.\\
\end{mechanism}

\begin{theorem}\label{thm:OptSubM}
Mechanism~\ref{mech:DSIC} uses $nm$ value queries and $n$ demand queries, and is DSIC for any CA with subadditive valuations. Moreover, if the valuations profile $\vv$ is $2$-stable submodular, Mechanism~\ref{mech:DSIC} returns the optimal allocation. 
\end{theorem}

\begin{proof}
First, we show that Mechanism~\ref{mech:DSIC} is DSIC for subadditive bidders. We focus on the bidding step, because assuming that each set $S_i$ is determined in a truthful way, it is always in each bidder's best interest to respond to his demand query truthfully. 
    
We observe that no bidder has incentive to bid lower than his singleton value for an item. Bidding lower could only lead to the bidder losing some items, thus restricting the set of items available to him through his demand query. Moreover, no bidder has incentive to bid higher than his singleton value for an item. This would only entail having access to an item that has price at least his actual singleton value. However, because of subadditivity, the bidder does not include such an item in his demand set. 

The fact that Mechanism~\ref{mech:DSIC} computes an optimal allocation for $2$-stable submodular valuations is an immediate consequence of Theorem~\ref{thm:OptSubA} and Theorem~\ref{thm:WalEq}. 
\end{proof}

Next, we show that a P2A (Mechanism~\ref{mech:P2A}), that uses only value queries, is ex-post incentive compatible when restricted to $2$-stable submodular valuations profiles. The proof of the following is an immediate consequence of Theorem~\ref{thm:WalEq} and Theorem~\ref{thm:OptSubA}.


\begin{mechanism}[bt]
\caption{Parallel 2nd Price Auction(P2A)}
\label{mech:P2A}
    \SetKwInOut{Input}{Input}
    \Input{Value query access to valuations $\vv = (v_1, \ldots, v_n)$.}
    \smallskip
    
    For all bidders $i$ and items $j$, query $v_i(j)$ and let $b_{ij}$ denote the response.\\
    Set the price $p_j$ of each item $j$ to its second highest bid.\\
    For each bidder $i$, let $S_i$ be the set of items for which $i$ has the highest bid.\\
    Bidder $i$ receives $S_i$ and pays $\sum_{j \in S_i} p_j$\,.\\
\end{mechanism}

\begin{theorem} \label{thm:P2AisEPIC}
    Mechanism~\ref{mech:P2A} uses $nm$ value queries and is EPIC for any CA with $2$-stable submodular valuations. Moreover, under truthful bidding, 
    Mechanism~\ref{mech:P2A} computes the optimal allocation. 
\end{theorem}

\begin{proof}
    The number of value queries immediately comes from the mechanism's description. The fact that it finds the optimal allocation comes directly from \cref{thm:OptSubA}; the same bundles are allocated, which was proven to be optimal even for subadditive bidders.


    To prove that the mechanism is EPIC, we fix a bidder $i$, his valuation $v_i(\cdot)$ and suppose that any other bidder $k$ bids his singleton value $v_k(j)$ for each item $j\in M$. We need to prove that $i$ has nothing to gain if he bids untruthfully.

    We will show this by contradiction. Suppose that $i$ makes a different bid and ends up with higher utility. In order to change his utility he must change either his allocated set or the payments. Since the payments are independent of his bids, he must to be allocated a set $S$ different than $O_i$. First we prove that it must hold that $S\sub O_i$.

    Let's assume the opposite, that there exists a $j$ for which $j\in S$ and $j\notin O_i$. Then $i$'s utility by dropping $j$ would decrease by $V = v_i(S)-v_i(S-j)$, but it would increase by $P = \max_{k\neq i}v_k(j)$. Because $j$ belongs to some $O_k$, by Lemma~\ref{LemmaGen}, and $2$-stability $P > v_i(j)$. Also because of submodularity, $V < v_i(j)$. This shows that $i$ has strictly higher utility if he drops any items not in $O_i$. Now we have to disprove that $S\subset O_i$.

    By contradiction, suppose that there exists a $j$ such that $j\in O_i$, but $j\notin S$. Then $i$'s utility by acquiring $j$ would increase by $V = v_i(S\cup j)- v_i(S)$ and decrease by $P = \max_{k\neq i}v_i(j)$. From Lemma~\ref{LemmaGen}, we know that $P < v_i(S\cup j) -v_i(S)$, which means that bidder $i$ gains strictly positive utility by obtaining any item in $O_i$.

    The theorem is now proven. We have shown that bidder $i$ strictly loses utility from items outside of $O_i$ and strictly gains utility from items in $O_i$. 
\end{proof}

Interestingly, Mechanism~\ref{mech:P2A} is not DSIC even when restricted to submodular CAs. The reason is that even the true bidder valuations profile may be $2$-stable, their bids might be not. Hence, it may happen that bidder $k$ bids higher than his real singleton value on some item $j$, but $j$ is allocated to different bidder $i$. This may increase $p_j$ to a level that is no longer profitable for bidder $i$ to get item $j$ (which is exactly the reason that we employ the demand queries in Mechanism~\ref{mech:DSIC}). 

The remark above naturally motivates the question about existence of a computationally efficient DISC mechanism that computes the optimal allocation for $2$-stable submodular CAs using only value queries. Rather surprisingly, the following answers this question in the negative. 

\begin{theorem}\label{thm:DSIC_no_exist}
    Let $\mA$ be any mechanism that is DSIC, uses only value queries and finds the optimal solution for $\g$-stable submodular valuations, for some $\g\geq 1$. Then $\mA$ makes exponentially many value queries.
\end{theorem}

\begin{proof}
The proof is an interesting adaptation of the proof of \cite[Theorem~3.1]{Dob11}. For the proof, we use instances with just $2$ bidders. Fixing one to be additive, the other may bid ``stably'' and get any bundle. However, due to the structure of his true valuation, finding his demand may be intractable, which makes misreporting a profitable strategy. 

To reach a contradiction, we assume that $\mA$ is DSIC, makes polynomially many value queries and always finds the optimal solution for $\g$-stable submodular valuations, for some fixed $\g\geq 1$. First we establish the following, which helps determining whether a set of additive valuations is stable.

\begin{proposition}\label{remark:stable_additive}
    Let $\vv$ be a profile with additive valuations. Then, for any $\g\geq 1$, $\vv$ is $\g$-stable if for any item $j \in M$, the largest value for $j$ in $\vv$ differs from the second largest value for $j$ in $\vv$ by a factor larger than $\gamma$. Namely, if $i =\arg\max_{k\in [n]}\{ v_k(j)\}$, then $v_i(j) > \g v_k(j)$, for all bidders $k \neq i$.
\end{proposition}

\begin{proof}[Proof of Proposition~\ref{remark:stable_additive}]
    The proposition follows directly from the fact that endowing an additive bidder $k$ for an item $j$ keeps him additive and inflates his singleton value by a factor of $\g$. 
\end{proof}

For the rest of the proof, we consider $2$ bidders and fix the valuation according to which bidder $1$ makes his bids: $v_1(S) = |S|/m$. $v_1(S)$ is not necessarily the true valuation of bidder $1$. We fix his true valuation to be also additive, with the value of each item large enough. This valuation, together with any other bounded and submodular valuation, results in a valuations profile that is submodular and stable (for a large enough stability factor). 

We next prove that bidder $1$ can get any bundle.

\begin{proposition}\label{remark:menu}
    For any bundle $O$, bidder 2 will be allocated $O$, if he bids according to
    \begin{equation}\label{eq:stable_additive}
        v_2(S) = |S\cap O| + \frac{|S - O|}{m^2}
    \end{equation}
\end{proposition}

\begin{proof}[Proof of Proposition~\ref{remark:menu}]
    First we fix any bundle $O \subseteq M$. By Proposition~\ref{remark:stable_additive}, valuations $(v_1, v_2)$ are $(m-\e)$-stable. Taking $m$ large enough makes the valuations $\g$-stable. Given that they are also additive, and thus submodular, we get that the mechanism $\mA$ (which, by hypothesis, computes the optimal allocation for $\gamma$-stable instances) allocates $O$ to bidder $2$. 
\end{proof}

Next, we show that the prices set by $\mA$ for bidder 2 are bounded and strictly increasing.

\begin{proposition}\label{remark:prices}
    For any bundle $T$ and any $S\sub T$, it holds that
    \begin{equation}\label{eq:prices}
        \frac{|T| - |S|}{m^2} \leq p_T - p_S \leq |T| - |S|
    \end{equation}
    where $p_S$ and $p_T$ are the prices of bundles $S$ and $T$ assigned from $\mA$ for bidder 2.
\end{proposition}

\begin{proof}[Proof of Proposition~\ref{remark:prices}]
    We examine what happens when bidder 2 bids according to \eqref{eq:stable_additive}. First we set $O=T$, which means that bidder 2 will receive $T$. Since $\mA$ is DSIC, bidder 2 should not prefer $S$ over $T$, i.e., $v_2(T) - p_T \geq v_2(S) - p_S$. This implies the rhs of \eqref{eq:prices}, because $v_2(T) = |T|$ and $v_2(S) = |S|$ (since $S\sub T$).

The argument for the lhs of \eqref{eq:prices} is symmetric. We let bidder 2 bid according to \eqref{eq:stable_additive}, where $O=S$. Then, bidder 2 should not prefer $T$ over $S$, i.e., $v_2(S) - p_S \geq v_2(T) - p_T$. Since $v(S) = |S|$ and $v(T) = m\cd|S| + |T-S|/m^2$, we get the lhs of \eqref{eq:prices}. 
\end{proof}

We also note that setting $S=\emptyset$ in \eqref{eq:prices}, we get that 
$|T|/m^2 \leq p_T \leq |T|$.

We now create an exponentially large structured submenu, as in \cite[Definition~3.2]{Dob11}. This concludes the proof, since the existence of such a submenu entails that $\mA$ requires exponentially many value queries to find the demand of bidder 2, as shown in \cite[Lemma~3.10]{Dob11}. For completeness, we recall that a collection of bundles $\mS$ comprises a \emph{structured submenu} for bidder 2 if:
\begin{enumerate}
    \item For all $S\in\mS$, bidder 2 can be allocated $S$.
    \item For each $S,T\in\mS$: $|S|=|T|$ and $|p_S - p_T| \leq \frac{1}{m^5}$.
    \item For all $S,T \subseteq M$ such that $S\in\mS$ and $S\subset T$: $p_T - p_S \geq \frac{1}{m^3}$.
    \item For all $S\in\mS$: $p_S\leq m$.
\end{enumerate}

Since bidder 2 can get any bundle, the first property is satisfied. Also by Proposition~\ref{remark:prices}, the third property is satisfied, because for any $S\subset T$: $p_T - p_S > 1/m^3$.

To create the structured submenu, we fix $k=m/2$ and consider all the $\binom{m}{k}$ different bundles of size $k$. Our submenu is a subset of those bundles, which immediately satisfies the first part of the second property. Also since $|T|/m^2 \leq p_T \leq |T|$, the price of each bundle is at most $k$, which implies the last property.

We need to show the last part of the second property. To this end, following the construction of \cite[Section~3.1]{Dob11}, we split the interval $[0,m]$ into $m^5$ bins. For each bundle $S$ of size $k$, we put $S$ in the $i$-th bin if $p_S \in \big[i/m^5, (i+1)/m^5\big)$. Since there are $m^5$ bins and $\binom{m}{k}$ bundles, one bin must have exponentially many bundles. Let $\mS$ be the set of bundles in such a bin. Notice that the bundles of the same bin have prices which differ less $1/m^5$, thus satisfying the last part of the second property. 

This completes the proof that there is an exponentially large collection $\mS$ of bundles that comprises a structured submenu. To conclude the proof, we apply \cite[Lemma~3.10]{Dob11}. For completeness, we include the technical details in Appendix~\ref{s:app:exponential}.
\end{proof}

A natural question is whether one could also follow the first part of the proof of \cite[Theorem~3.1]{Dob11}, in order to get a much stronger  inapproximability bound of $m^{-1/2+\e}$, for any $\e > 0$. Unfortunately the answer is negative, because the polar additive valuation profiles in \cite[Section~3.1]{Dob11} are far from stable. This explains the necessity of our careful construction of stable valuation profiles, in the first part of the proof of Theorem~\ref{thm:DSIC_no_exist}.


\section{Price of Anarchy in Stable Combinatorial Auctions}
\label{s:poa}

\subsection{The Price of Anarchy in Parallel 2nd Price Auctions}

For XOS valuations, the PoA of P2A is at least $1/2$ \cite{RST16}. We next show that even for $(3-\e)$-stable valuations, the PoA of P2A does not improve.

\begin{lemma}\label{lem:PoA_counter_example}
There exists a $(3-\e)$-stable profile with unit-demand valuations for which the PoA of P2A is $1/2$.
\end{lemma}

\begin{proof}
The instance in Example~\ref{ex:AliceBob} is $(3-\e)$-stable and has been used to show that the PoA of P2A is at most $1/2$. More precisely, we observe that there is an equilibrium with social welfare $2$: Alice bids $0$ for $a$ and $1$ for $b$ and Bob bids $1$ for $a$ and $0$ for $b$.
\end{proof}

Interestingly, the previous result is tight. With $3$-stable submodular  valuations, every equilibrium is optimal. To prove this, we introduce a no-overbidding assumption, which is weaker than the usual Strong No Overbidding assumption (where each bidder's value for any bundle $S$ is at most the sum of his bids for $S$). We call our assumption Singleton No Overbidding (SiNO), as it restricts the bids to be below each corresponding singleton value. We also note that the bidding profile in Lemma~\ref{lem:PoA_counter_example} SiNO. 

\begin{definition}[Singleton No Overbidding]
    A bidding profile $(\vb_1,...,\vb_n)$ satisfies \textit{Singleton No Overbidding (SiNO)} if for any bidder $i$ and item $j$: $v_i(j) \geq b_{ij}$.
\end{definition}

%
Now we are ready to prove that with SiNO, PoA is always $1$ for CAs with $3$-stable submodular valuations.

\begin{theorem}\label{theoremP2AStable}
Let $\vv$ be a $3$-stable submodular valuations profile, and let $\vb$ be a bidding profile that forms a Pure Nash Equilibrium for P2A and satisfies SiNO. Then the allocation at the equilibrium coincides with the optimal allocation.
\end{theorem}


\begin{proof}
    We prove the theorem by contradiction. First, denote with $w_j$ the second highest singleton valuation for item $j$, i.e. if $j\in O_i$, $w_j=\max_{k\neq i}v_k(j)$ ($i$ has the highest singleton valuation for $j$, by Lemma~\ref{LemmaGen}). Fix a bidder $i$ who does not get allocated his optimal set, i.e. $O_i\neq S_i(\vb)$. Let $A_i$ be the items that $i$ is allocated at the equilibrium that are also in $O_i$, i.e. $A_i = O_i\cap S_i(\vb)$.
    
    We construct a deviating bid for $i$: For items not in $O_i$, $i$ bids $0$. For items in $A_i$, $i$ bids as before (which means that the prices for these items will be the same as before). Finally for items in $O_i-A_i$, $i$ bids infinitesimally more than $w_j$. Note that this bidding strategy conforms to SiNO. Also because of SiNO, it guarantees that $i$ receives the whole bundle $O_i$. His new utility is

    \begin{equation}\label{ineq5:12}
        u_i' = v_i(O_i) - \sum_{j\in A_i}p_j(\vb) - \sum_{j\in O_i-A_i}\max_{k\neq i}b_{kj}   \geq
               v_i(O_i) - \sum_{j\in A_i}p_j(\vb) - \sum_{j\in O_i-A_i}w_j\,,
    \end{equation}
    where the inequality holds because of SiNO: Every bid must be below the corresponding singleton value, which in turn is less than the maximum of the singleton values. Now we use that $i$' utility in \eqref{ineq5:12} cannot be greater than $i$'s utility at the equilibrium (for brevity, we let $S_i(\vb) = S_i$ from that point on). Hence, we obtain that: 
    \begin{equation}\label{ineq5:13}
        v_i(S_i) - \sum_{j\in S_i}p_j(\vb) \geq v_i(O_i) - \sum_{j\in A_i}p_j(\vb) - \sum_{j\in O_i-A_i}w_j
    \end{equation}

In \eqref{ineq5:13}, we can eliminate the term $\sum_{j\in A_i}p_j(\vb)$, because it is included in the term $\sum_{j\in S_i}p_j(\vb)$. Also, we can ignore the rest of the prices in the lhs, as they are positive. Hence, we obtain that:
    \begin{equation}\label{ineq5:14}
        v_i(S_i) \geq v_i(O_i) - \sum_{j\in O_i-A_i}w_j
    \end{equation}

    Now using that $w_j\geq v_i(j)$, for items $j\in S_i-A_i$, and submodularity, namely that $\sum_{j\in S_i-A_i}v_i(j) + v_i(A_i) \geq v_i(S_i)$, and \eqref{ineq5:14}, we obtain the following:
    \begin{equation}\label{ineq5:15}
        \sum_{j\in S_i-A_i}w_j + \sum_{j\in O_i-A_i}w_j \geq v_i(O_i) - v_i(A_i)
    \end{equation}

    Now because, by definition, both sets $\bigcup_i(S_i-A_i)$ and $\bigcup_i(O_i-A_i)$ consist of the items that were not optimally allocated (and only them), summing up \eqref{ineq5:15} over all bidders $i$, we get that:
    \begin{equation}\label{ineq5:16}
        2\sum_i\sum_{j\in O_i-A_i}w_j \geq \sum_i\big(v_i(O_i) - v_i(A_i)\big) 
    \end{equation}

    Because of Lemma~\ref{LemmaGen} and $3$-stability, for any item $j\in O_i$, we have that $v_i(j|O_i-j) > \frac{1}{2}w_j$. Then, using \eqref{ineq5:16}, we obtain the following:
    \begin{equation}\label{ineq5:17}
        \sum_i\sum_{j\in O_i-A_i}\big(v_i(O_i) - v_i(O_i-j)\big) > \sum_i\big(v_i(O_i) - v_i(A_i)\big) 
    \end{equation}

    For every $i$, it holds that $\sum_{j\in O_i-A_i}\big(v_i(O_i) - v_i(O_i-j)\big) \leq v_i(O_i) - v_i(A_i)$, which contradicts \eqref{ineq5:17}. The latter inequality is true because of submodularity. More specifically, we can index items in $O_i-A_i$ as $1,2,3,\ldots,|O_i-A_i|$ and denote $X_j=\{1,2,...,j\}$, with $X_0=\emptyset$ and $X_{|O_i-A_i|}=O_i-A_i$. Hence, we obtain that:
    \begin{equation}\label{eq5:100}
        v_i(O_i) - v_i(A_i) = \sum_{j=1}^{|O_i-A_i|}\big(v_i(O_i-X_{j-1}) - v_i(O_i-X_j)\big)
    \end{equation}
    
    Due to submodularity, $v_i(O_i-X_{j-1}) - v_i(O_i-X_j) \geq v_i(O_i) - v_i(O_i-j)$. This and \eqref{eq5:100} imply that $\sum_{j\in O_i-A_i}\big(v_i(O_i) - v_i(O_i-j)\big) \leq v_i(O_i) - v_i(A_i)$, which concludes the proof of the theorem.
\end{proof}

We proceed to study the PoA of P2A for the more general class of XOS valuations. For general XOS valuations, the PoA is at least $1/2$, which cannot be improved for $(3-\e)$-stable XOS valuations, due to Lemma~\ref{lem:PoA_counter_example}. We show that as valuations become more stable, the PoA improves.

\begin{theorem}\label{theoremStableP2AXOS}
For any $\g\geq 2$, let $\vv$ be a $\g$-stable profile with XOS valuations. Let $\vb$ be a bidding profile that forms a Pure Nash Equilibrium for P2A and satisfies SiNO. Then the PoA is larger than $\frac{\g-2}{\g-1}$.
\end{theorem}

The intuition is similar to that in the proof of Theorem~\ref{theoremP2AStable}. Even if the prices are as high as possible, as $\g$ gets larger, each bidder has more incentive to prefer the items in his optimal bundle than any other items.

\begin{proof}
    Fix the optimal allocation $(O_1,...,O_n)$ and a bidder $i$. Let $q_i(\cdot)$ be the supporting additive valuation for bidder i for set $O_i$: $\forall S, v_i(S)\geq\sum_{j\in S}q_i(j)$ and $v_i(O_i)=\sum_{j\in O_i}q_i(j)$. We have bidder $i$ deviate and bid according to $\vec{b_i^*}$: he bids $0$ for items not in $O_i$, and for $j\in O_i$ he bids $q_i(j)$.
	
	 Because of SiNO ($v_k(j) \geq b_{kj}$) and Lemma~\ref{LemmaGen} ($q_i(j) > (\g-1)\cd\max_{k\neq i}v_k(j)$), bidder is guaranteed to obtain the whole bundle $O_i$.
    Because $\vb$ is a Nash Equilibrium it holds that

    \begin{equation} \label{eq66}
        u_i(S_i(\vb)) \geq u_i(S_i(\vec{b_i^*},\vec{b_{-i}}))   =
        v_i(O_i) - \sum_{j \in O_i} \underset{k\neq i}{max}(b_{k j})
    \end{equation}

    Since prices are non-negative, lhs in (\ref{eq66}) is at most $v_i(S_i(\vb))$ and because of SiNO $b_{k j} \leq v_k(j)$. Also because of Lemma~\ref{LemmaGen}: $(\g-1)\cdot v_k(j) < q_i(j)$. Combining all the above in inequality (\ref{eq66}) gives:

    \begin{equation} \label{eq77}
        v_i(S_i(\vb)) > v_i(O_i) - \frac{1}{\g-1}\sum_{j\in O_i}q_i(j)   =
        \frac{\g-2}{\g-1} \; v_i(O_i)
    \end{equation}

    Summing inequality (\ref{eq77}) for all $i$ gives us the lemma and completes the proof.
\end{proof}

\subsection{The Price of Anarchy in Parallel 1st Price Auctions}

%
We conclude with a lower bound on the PoA of Parallel 1st Price Auctions (P1A) for CAs with stable valuations. If bidders are restricted to a mixed Nash equilibrium, the PoA of P1A for bidders with XOS valuations is at least $1-\fe$. Similarly to Theorem~\ref{theoremStableP2AXOS}, we show that the PoA of P1A increases, as the stability of a XOS valuations profile increases.

\begin{theorem}\label{theoremStableP1A}
For any $\g\geq 2$, let $\vv$ be a $\g$-stable profile with XOS valuations. Let $\vb$ be a bidding profile that forms a Mixed Nash Equilibrium for P1A. Then the PoA is larger than $\frac{\g-2}{\g-1}$.
\end{theorem}

For the proof, we observe that as the stability factor $\g$ increases, the valuation of a bidder $i$ for each item $j$ in his optimal bundle becomes considerably larger than the second highest singleton valuation for item $j$. Hence, if bidder $i$ bids the second highest singleton valuation for each item in his optimal bundle, $i$'s utility should be large enough to establish that the allocation of the equilibrium achieves a large enough social welfare.

Before proving Theorem~\ref{theoremStableP1A}, we need to prove a technical proposition. Because in stable auctions the singleton values of the bidders are very important, we would like to prove that in a P1A, no bidder bids higher than his singleton value, like in SiNO. However this is not the case. One can consider a simple example with $2$ bidders and $1$ item. The first bidder has value $1$ and the second $2$. A pure equilibrium is for the first bidder to bid $1.5$ and the second bidder to bid $1.5+\e$.

Since bids can be above singleton values, we are going to prove something similar: A bidder bids higher that his singleton value, only if he is sure that he is not going to receive that item.

\begin{proposition}\label{ClaimWithProbs}
    Let $\vv$ be subadditive valuations and $\vD = (D_1,...,D_n)$ a profile of bid vector distributions that forms a MNE for P1A. Then for every $j$ and $i$:

    \begin{equation*}
        \textit{If } \Pr_{\vb\sim \vD}[j\in S_i(\vb)]>0, \textit{ then } \max_{\vb_i\sim D_i} (b_{ij}) \leq v_i(j)
    \end{equation*}
\end{proposition}

Proposition~\ref{ClaimWithProbs} states that if a bidder gets a certain item with positive probability, then he always bids at most his singleton value for that item. The proof is quite easy, since paying for an item more than the corresponding singleton value, decrease the utility.

\begin{proof}[Proof of Proposition~\ref{ClaimWithProbs}]
    Fix a bidder $i$ and an item $j$. We will show this by contradiction: Suppose that at some realization of $D_i$ bidder $i$ bids higher than $v_i(j)$ for item $j$ and that at some other realization of $\vD$ bidder $i$ gets item $j$. Since the distributions in $\vD$ are independent, this means that there exists a realization of $\vD$ where bidder $i$ both gets $j$ and bids for it higher than $v_i(j)$. Now we need to show that if bidder $i$ lowered his bid to $v_i(j)$ he would gain utility.

    In the realizations where $i$ bids high for $j$ and he does not get $j$, he has nothing to lose by lowering his price.

    In the realizations where $i$ bids high for $j$ and gets $j$, if he lowered his price to $v_i(j)$ he would either lose the item, yielding a utility increase of $b_{ij} - v_i(j|S-j)$ or he would simple pay less, increasing his utility by $b_{ij}-v_i(j)$. Both quantities are strictly positive, the first because of subadditivity ($v_i(j|S-j)<v_i(j)$). This means that $\vD$ is not a MNE, which completes the contradiction.
\end{proof}

Having proven Proposition~\ref{ClaimWithProbs}, we are ready to show our main result for stable P1A. From now on we are going to denote with $w_j$ the maximum of the singleton values of the bidders that do not get item $j$ at the optimal allocation, i.e. if $j\in O_i$, then $w_j = \max_{k\neq i}v_k(j)$.

\begin{proof}[Proof of Theorem~\ref{theoremStableP1A}]
    Fix a bidder $i$. We are going to construct a deviating bid for him. First bid $0$ for any items not in $O_i$. Denote with $A$ the subset of $O_i$ that $i$ gets with probability $1$ at the MNE. For the items in $A$, he keeps the same bidding strategy and for the other items $j\in O_i-A$ he bids infinitesimally more than $w_j$, i.e. $\max_{k\neq i}v_k(j)$. This bidding strategy gets $i$ the whole bundle $O_i$. This is obvious for items in $A$. For items $j\in O_i-A$, player $i$ only needs to outbid any other player $k$ that has positive probability to get $j$. Because of Proposition~\ref{ClaimWithProbs} their bid is at most $w_j$, which means that our deviating bids achieve getting these items from them.

    Before using our deviating bid, let us analyze the expected payment of $i$ at the equilibrium:

    \begin{equation}\label{ineq5:7}
        \E_{\vb\sim\vD}\sum_{j\in S_i(\vb)}p_j(\vb)   =
        \E_{\vb\sim\vD}\sum_{j\in S_i(\vb)-A}p_j(\vb) + \E_{\vb\sim\vD}\sum_{j\in A}p_j(\vb)   \geq
        \E_{\vb\sim\vD}\sum_{j\in A}p_j(\vb)
    \end{equation}

    Where the equality holds because for every realization of bids $\vb$, $A\sub S_i(\vb)$ and the inequality because the payments are always non-negative. Now we use the fact that if $i$ uses the deviating bids, he is not going to lose any utility. For the payment at the equilibrium we use \eqref{ineq5:7} as it is a lower bound for the real payment of $i$

    \begin{equation}
        \E_{\vb\sim\vD} v_i(S_i(\vb)) - \E_{\vb\sim\vD}\sum_{j\in A}p_j(\vb)   \geq
        v_i(O_i) - \sum_{j\in O_i-A}w_j - \E_{\vb\sim\vD}\sum_{j\in A}p_j(\vb)
    \end{equation}

    Now by rearranging and using Lemma~\ref{LemmaGen} ($w_j<\frac{1}{\g-1}q_i(j)$) we get

    \begin{equation}\label{ineq5:9}
        \E_{\vb\sim\vD} v_i(S_i(\vb))   >
        v_i(O_i) - \frac{1}{\g-1}\sum_{j\in O_i-A}q_i(j)   \geq
        v_i(O_i) - \frac{1}{\g-1}\sum_{j\in O_i}q_i(j)
    \end{equation}

    Where in the second inequality we simply make the sum contain more positive terms. Using now the fact that $v_i(O_i) = \sum_{j\in O_i}q_i(j)$ and by adding \cref{ineq5:9} for all $i$, we complete the proof of the theorem.
\end{proof} 


\smallskip\noindent{\bf Acknowledgements.} We wish to thank Kyriakos Lotidis and Grigoris Velegkas for many helpful discussions on combinatorial markets with endowed valuations and on the possibility of exploiting endowed valuations in mechanism design.

\bibliographystyle{plain}
\bibliography{references}

\appendix\section{Appendix}

\subsection{Technical Details Missing from the Proof of Theorem~\ref{thm:DSIC_no_exist}}
\label{s:app:exponential}

In this section, we complete the proof of Theorem~\ref{thm:DSIC_no_exist}, working along the lines of the proof of \cite[Lemma~3.10]{Dob11}. We have already showed that there exist a structured submenu of exponential size with the following properties:
\begin{enumerate}
    \item For all $S\in\mS$, bidder 2 can be allocated $S$.
    \item For each $S,T\in\mS$: $|p_S - p_T| \leq \frac{1}{m^5}$.
    \item For all $S,T \subseteq M$ such that $S\in\mS$, $S\subset T$: $p_T - p_S \geq \frac{1}{m^3}$.
    \item For all $S\in\mS$: $p_S\leq m$.
    \item For each $S,T\in\mS$: $|S|=|T|$.
\end{enumerate}

With this submenu, we are ready to create a valuation where it is a hard to find the demand bundle. Given a set $O\in \mS$, we use, as in the proof of \cite[Lemma~3.10]{Dob11}, the following valuation function: 
\begin{equation}\label{eq:magic}
    v^O(S) =
        \begin{cases}
            t\cd|S|, &\textit{if } |S| < k \\
            t\cd k - 1/m^4, &\textit{if } S\in \mS \textit{ and } S\neq O\\
            t\cd k, &\textit{if } S=O \textit{ or } \exists T\in \mS \textit{ s.t. } T\subset S\\
            t\cd(k-\frac{1}{2^{m-|S|}}), &\textit{otherwise}\\
        \end{cases}
\end{equation}
where $t$ is a large number whose value we will set later. To establish the submodularity of $v^O$, we show that for any $S\subset T$ and $j\notin T$, we have that $0 \leq v(S\cup j) - v(S) \geq v(T\cup j) - v(T)$. Since each marginal is at most $t$, any case where $v(S\cup j) - v(S) = t$ is trivial. Now we notice:
\begin{itemize}
    \item If $S\cup j \in \mS$, then $v(S\cup j) - v(S) \geq t-1/m^4$. However the value of the marginal $v(T\cup j) - v(T)$ is one of the following:
    \begin{itemize}
        \item $1/m^4$, if $T\in \mS-\{O\}$.
        \item $0$, if $T = O$.
        \item $\frac{t}{2^{m-|T|}}$, if both $v(T)$ and $v(T\cup j)$ are calculated by the 4th case of the valuation.
        \item $\frac{t}{2^{m-|T|}}$, if $v(T)$ is calculated by the 4th case of the valuation and $v(T\cup j)$ is calculated by the 3rd case of the valuation.
    \end{itemize}
    Because $|T|<m$, in all cases $t-1/m^4$ is greater for a large enough $t$, e.g. $t>2^m$.

    \item If $\exists S'\in \mS \textit{ s.t. } S'\subset S$, then $v(S\cup j) - v(S) = v(T\cup j) - v(T) = 0$.

    \item If $v(S) = t\cd(k-\frac{1}{2^{m-|S|}})$, then $v(S\cup j) - v(S) = \frac{t}{2^{m-|S|}}$. Because $S\sub T$, the value of $v(T\cup j) - v(T)$ is either going to be $0$ or $\frac{t}{2^{m-|T|}}$. In any case the marginal of $S$ is greater.
\end{itemize}

Thus we see that in any case the marginal of $S$ is both positive and greater than the marginal of $T$. This proves the submodularity and monotonicity.

Now we notice that finding bundle $O$ by doing value queries in $v^O(\cd)$ is hard: it requires $|\mS|-1$ queries. This is because for any $O'\in\mS$ distinguishing between $v^O(\cd)$ and $v^{O'}(\cd)$ requires querying either $O$ or $O'$, which in the worst case requires querying almost every bundle in $\mS$. Since $\mA$ does polynomially many value queries, bidder 2 will not be allocated bundle $O$, since the size of $\mS$ is exponential.

Now all that is left to do is to prove is that if bidder 2, who can pick any set $S$ at price $p_S$, has valuation $v^O(\cd)$, then he strictly demands $O$. To this end, we need to prove the following inequality (see also \cite[Claim~3.13]{Dob11}):
\begin{equation}\label{ineq:utility}
    v^O(O) - p_O > v^O(S) - p_S
\end{equation}
for any $S\neq O$. We examine the following cases:
\begin{itemize}
    \item If $|S| < k$ then $v(S) = t\cd|S|$. Manipulating inequality \ref{ineq:utility} leads to proving $t\cd(k-|S|) > p_O - p_S$. The lhs is at least $t$ (since $|S|\leq k-1$) and the rhs is at most $m$ (because of submenu property 4 and $p_S \geq 0$). Since $t$ is as large as we want it to be, it holds that $t > m$.
    \item If $S\in \mS$ then inequality \ref{ineq:utility} becomes $1/m^4 > p_O - p_S$. This indeed holds because of submenu property 2.
    \item If $O\subset S$, then $v(S) = v(O)$. However because of submenu property 3, the price of $S$ is greater than $O$'s. This makes it have strictly less utility.
    \item If $\exists T\in \mS$ s.t. $T\subset S$ ($T\neq O$), then $v(S) = v(T) + 1/m^4$ and also $p_S \geq p_T + 1/m^3$ (submenu property 3). This makes $S$ at less favorable than $T$. However $T$ already has less utility than $O$ because of the second bullet.
    \item If $v(S) = t\cd(k-\frac{1}{2^{m-|S|}})$, then inequality \ref{ineq:utility} becomes $\frac{t}{2^{m-|S|}} > p_O - p_S$. This holds because we can take $t$ as large as we want, making the LHS as large as we want, while at the same time the RHS is at most $m$ (because of submenu property 4 and $p_S \geq 0$).
\end{itemize}

Thus we have proven \eqref{ineq:utility}, for every $S\neq O$. This makes bidder 2 strictly demand the bundle $O$. However the mechanism cannot allocate that bundle with polynomially many queries. This entails that the best strategy for bidder 2 is to bid untruthfully, e.g. bidding according to $v(S) = |S\cap O|$. This makes $\mA$ not DSIC which leads to a contradiction and completes the proof.

\end{document}